\documentclass[pra,twocolumn,longbibliography]{revtex4-2}

\usepackage{braket}
\usepackage[T1]{fontenc}
\usepackage{amsmath}
\usepackage{amsthm}
\usepackage{amssymb}
\usepackage{amsfonts}
\usepackage{mathtools}
\usepackage{dsfont}
\usepackage[normalem]{ulem}
\usepackage{graphicx}
\usepackage{siunitx}
\usepackage{mhchem}
\usepackage{datetime}
\usepackage{enumitem}   
\usepackage{color}
\usepackage{newtxtext,newtxmath}
\usepackage{microtype}
\usepackage{braket}
\usepackage{xr}
\usepackage{natbib}
\usepackage{hyperref}
\hypersetup{%
    bookmarksopen=false,
    bookmarksnumbered=true,
    pdfnewwindow=true,
    unicode=false,
    colorlinks=true,%
    citecolor=blue,
    linkcolor=black,
    urlcolor=blue,
    filecolor=blue
    }

\newcommand{\bcdot}{\boldsymbol{\cdot}}

\newtheorem{theorem}{Proposition}[section]
\theoremstyle{definition}

\newcommand{\dket}[1]{|#1\rangle\!\rangle} 
\newcommand{\dbra}[1]{\langle\!\langle#1|} 
\newcommand{\dbraket}[1]{\langle\!\langle#1\rangle\!\rangle} 
\newcommand{\tr}{\mathrm{tr}} 
\begin{document}
\newcommand{\figdir}{figures}
\newcommand{\liege}{Institut de Physique Nucléaire, Atomique et de Spectroscopie, CESAM, Universit\'e de Li\`ege, 4000 Liège, Belgium}
\title{Spectral theory of non-Markovian dissipative phase transitions}

\author{Baptiste Debecker}
\affiliation{\liege}
\author{John Martin}
\affiliation{\liege}
\author{François Damanet}
\affiliation{\liege}

\begin{abstract}
Dissipative phase transitions in quantum systems have been largely studied under the so-called Markovian approximation, where the environments to which the systems are coupled are memoryless. Here, we present a generalization of the spectral theory of dissipative phase transitions to non-Markovian systems, encompassing a much broader class of quantum materials and experiments and opening many possibilities for non-Markovian engineering of matter phases such as, as explored in the companion Letter~[Debecker \emph{et. al.}, Phys.\ Rev.\ Lett.\ \textbf{133}, 140403 (2024)], reshaping of phase boundaries and triggering of phase transitions. We first prove several statements about the connections between the spectrum of the generator of the non-Markovian dynamics of general systems and dissipative phase transitions. Then, as a benchmark, we show that our framework can capture \textit{all} the expected signatures of the superradiant phase transition appearing in a challenging $\mathbb{U}(1)$-symmetric two-mode Dicke model from a reduced description of the dynamics of the atoms only, a task for which all other methods have failed so far. 
\end{abstract}
\date{\today}
\maketitle

\section{Introduction}

Understanding how phase transitions emerge in quantum systems is one of the most important forefront in physics, as it could help finding new routes to fabricate superconductors at higher temperatures, to cite but one. In this context, exploring how driven-dissipative mechanisms obtained via the coupling of the quantum systems to engineered environments and fields constitute an interesting research direction, as it has been shown this could yield to matter phases otherwise inaccessible ~\cite{Jin2016Cluster, Lee2013Unconventional}, such as phases with long-range order in 2D~\cite{Toner1998Flocks}, which is forbidden at equilibrium~\cite{Mermin1966}.

Phase transitions triggered out of equilibrium, i.e., dissipative phase transitions (DPTs), have been observed in various experiments~\cite{Baumann2010Apr, Klinder2015Dynamical, Fitzpatrick2017Observation, Fink2017Observation, Fink2017Signatures, Benary2022}. On the theory side, DPTs have mostly been studied under the so-called Markovian approximation, i.e., when quantum systems are coupled to memoryless reservoirs~\cite{Minganti2018Spectral, Kessler2012Dissipative, Hwang2018Dissipative}. One of the main reasons explaining this is that writing the equations of motion in the Markovian limit does not require knowing the spectrum of the quantum systems, a task which is particularly difficult to achieve for large, many-body systems. 

However, most realistic experimental platforms and quantum materials are coupled to reservoirs with finite memory~\cite{deVega2017Dynamics}, inducing non-Markovian effects which can lead to interesting phenomena~\cite{Huelga2012Non, Maier2019Environment, Link2022, Chen2023C, Otterpohl2022Hidden}. Non-Markovian effects also appear when one derives reduced descriptions of a large Markovian open quantum system to deal with a smaller Hilbert space~\cite{Damanet2019, Palacino2020, Link2022}. A paradigmatic example for this are atoms trapped in a lossy cavity QED, where the atoms and the cavity modes constitute an enlarged Markovian system coupled to the outside electromagnetic field. In these systems, one is usually mostly interested in the behavior of the atomic dynamics, and would thus strongly benefit, also from a computational perspective, from a reduced description of it, especially in the multimode cavity case for which the size of the Hilbert space becomes quickly intractable~\cite{Kollar2015, Gopalakrishnan2009, Gopalakrishnan2010, Gopalakrishnan2011, Kollar2017Feb, Torggler2019quantumnqueens, Vaidya2018Tunable, Guo2019Emergent, Guo2019Sign, Marsh2021Enhancing, Fiorelli2020}. To derive reduced descriptions of atoms in this context, one usually relies on the adiabatic elimination of the cavity modes. However, it has been shown that precautions must be taken in the approximations performed in the models to capture the correct behavior of the system, as compared for the superradiant phase transitions appearing in driven-dissipative Dicke models~\cite{Damanet2019, Jager2022Lindblad, Palacino2020}. This highlights the need for a systematic framework to characterize DPTs in arbitrary systems.

Here, we present a general method to characterize signatures of DPTs in non-Markovian systems, opening possibilities for exploring DPTs in a wider range of setups. Our approach is based on the Hierarchical Equations of Motion (HEOM)~\cite{Tanimura89, Ishizaki2005, Ishizaki2009, Yoshitaka2020, Lambert2020Qutip,Huang2023heomjl}, a numerical method for non-Markovian dynamics widely used in quantum physics and chemistry. From these equations of motion, one can define a generalization of the Liouvillian usually associated with the Lindblad master equation for Markovian systems, whose spectral properties are connected to DPTs. One of the necessary conditions for DPTs is the closing of the Liouvillian gap~\cite{Minganti2018Spectral}: Here we show how to define a similar quantity for non-Markovian systems.

Other techniques to study the effects of non-Markovianity on DPTs have been used, such as Green functions to study the
impact of the environment spectral density on critical
exponents~\cite{Nagy2015Nonequilibrium, Nagy2016Critical}, Lindblad master equations with time-dependent
rates to characterize the dynamics of a probe
coupled to a non-Markovian environment~\cite{Haikka2012NonMarkovianity}, or time evolving
matrix product operators (TEMPO) to localize
DPT in quantum systems coupled to a non-Markovian harmonic environment~\cite{Strathearn2018Efficient}. Most of them focus however on the paradigmatic spin-boson model~\cite{Anders2007, Chin2011Generalized, Florens2010Quantum}. As our approach is the natural extension of the powerful spectral machinery widely used for Markovian systems, it provides an ideal framework to explore non-Markovian effects in new regimes and for more realistic systems and experiments.

This paper presents all the technical details of our spectral theory of non-Markovian dissipative phase transitions as well as a few benchmarks. It goes along the companion Letter~\cite{Debecker2024PRL} which focuses on the application of the framework to highlight new physics and opportunities for dissipation engineering of DPTs, as it shows how non-Markovian effects can be used to reshape phase boundaries but also trigger phase transitions where Lindblad master equations do not.

The paper is organized as follows. In Sec.~II, we first present the generalization of the Liouvillian for non-Markovian systems, i.e., the HEOM Liouvillian, and its properties. In Sec.~III, we present our spectral theory of non-Markovian DPTs, by demonstrating how the spectral properties of the HEOM Liouvillian connect to DPTs and symmetries. In Sec.~IV, we illustrate how our framework can be used to understand matter phases from the decomposition of the steady states for the $1^{\mathrm{st}}$- and $2^{\mathrm{nd}}$-order DPT addressed in our companion Letter~\cite{Debecker2024PRL}. Finally, in Sec.~V, we show that our theory can be used to capture all the expected features of a $2^{\mathrm{nd}}$-order DPT associated with a continuous symmetry in a challenging two-mode Dicke model for which all previous reduced descriptions had failed so far~\cite{Palacino2020}.

\section{Generator of non-Markovian dynamics}

In this Section, we first introduce the generator of non-Markovian dynamics that is the centerpiece of our spectral theory of DPTs. We then discuss its properties and finally the computational advantage of using it over standard Markovian embedding of non-Markovian systems.

\subsection{HEOM Liouvillian}
We consider an open quantum system $S$ of Hamiltonian $H_S$ interacting linearly through $H_\mathrm{int}$ with $N_E$ environments of Hamiltonian $H_E$ made of a collection of bosonic modes. Setting $\hbar = 1$, the total Hamiltonian reads
\begin{equation}\label{Htot}
 \begin{split}
    H &= H_S + H_E + H_\mathrm{int}\\
    &= H_S + \sum_{l = 1}^{N_E}\sum_{k} \omega_{l k} b^\dagger_{l k} b_{l k} + \sum_{l = 1}^{N_E}\sum_{k} \left(g_{l k} b_{l k} L_l^\dagger + g^*_{l k} b_{l k}^\dagger L_l\right),
 \end{split}
\end{equation}
where $b_{l k}$ ($b_{l k}^\dagger$) is the annihilation (creation) operator associated with the $k$-th bosonic mode  of frequency $\omega_{l k}$ of the $l$-th bath and $g_{l k}$ is a coupling constant characterizing the strength of the coupling between the system and the mode $k$ of the bath $l$. For each bath, the interaction is mediated by an arbitrary system operator $L_l$.

At zero temperature, environments are characterized by correlation functions $\alpha_l(\tau)$ of the form
\begin{equation}
    \alpha_l(\tau) = \sum_k |g_{l k}|^2 e^{-i\omega_{l k} \tau}.
\end{equation}
Note however that the formalism presented here is easily generalizable to bosonic or fermionic
baths at finite temperatures, by employing the so-called thermofield method~\cite{Semenoff1983}, for example.
Now, we suppose that each correlation function $\alpha_l$ can be written, either exactly or approximately~\cite{Meier1999Non,Ritschel2014, Hartmann2019, Lambert2020Qutip}, as a sum of $M_l$ decaying exponentials, i.e., 
\begin{equation}
    \alpha_l(\tau) =  \sum_{j=1}^{M_l} G_{l j}\, e^{-i\omega_{l j} \tau - \kappa_{l j} |\tau| },\quad \kappa_{l j}, \omega_ {l j} \in \mathbb{R},\; G_{l j} \in \mathbb{C}.
    \label{CF}
\end{equation}
This decomposition, usually performed with $G_{lj} \in \mathbb{R}$, is the so-called pseudomode picture~\cite{Imamoglu1994Stochastic, Dalton2001Theory, Garraway1997Nonperturbative, Pleasance2020Generalized, Mazzola2009Pseudomodes, Yang2012Nonadiabatic,Breuer2004Genuine, Barchielli2010Stochastic, Link2022}. It amounts to decomposing each environment $l$ into a set of $M_l$ independent modes of frequencies $\omega_{lj}$ damped with rates $\kappa_{lj}$, which is relevant for many experimental setups such as atoms in cavities, superconducting qubits coupled to resonators~\cite{Schmidt2013CircuitQED, Blais2021Circuit}, electrons-phonons systems~\cite{Flannigan2022Manybody, Moroder2022Strongly}, or emitters in plasmonic cavities~\cite{Santhosh2015Vacuum}.

A common strategy to compute the dynamics of the non-Markovian system $S$ coupled to such damped pseudomodes consists of including the pseudomode degrees of freedom in the system description. This defines an enlarged Markovian system $S_M$ described by a density operator 
$\rho_\mathrm{tot}$ whose dynamics is governed by a standard Markovian master equation of the form
\begin{equation}\label{MEMArkov}
\frac{d\rho_\mathrm{tot}}{dt} = - i \left[H_\mathrm{tot}, \rho_\mathrm{tot}\right]  + \sum_{l= 1}^{N_E}\sum_{j = 1}^{M_l}\kappa_{lj} \mathcal{D}_{a_{lj}}[\rho_\mathrm{tot}] \equiv \mathcal{L}_\mathrm{M}[\rho_\mathrm{tot}],
\end{equation}
where 
\begin{equation}
    H_\mathrm{tot} = H_S +  \sum_{l= 1}^{N_E}\sum_{j = 1}^{M_l}\left[ \omega_{lj} a_{lj}^\dagger a_{lj} + (G_{lj}a^\dagger_{lj} L_{l} + G_{lj}^* L_{l}^\dagger a_{lj}) \right]
\end{equation}
is the ``system + pseudomodes'' Hamiltonian and $\mathcal{D}_o[\boldsymbol{\cdot}] = 2o\boldsymbol{\cdot} o^\dagger - \{o^\dagger o, \boldsymbol{\cdot}\}$ the dissipator inducing the damping of the pseudomodes, where $a_{lj}$ is the annihilation operator of the $j^{\mathrm{th}}$ pseudomode of the bath $l$. The superoperator $\mathcal{L}_\mathrm{M}$ defined in Eq.~(\ref{MEMArkov}) is a standard Markovian Liouvillian and constitutes the generator of the dynamics of the Markovian system $S_M$. The connections between its spectral properties and dissipative phase transitions have been studied in~\cite{Minganti2018Spectral}. 

In this paper, we employ a different approach and derive a spectral theory of non-Markovian DPTs directly based on the generator of the dynamics of the non-Markovian system $S$ described below. When the global system is initially in the state $\rho_\mathrm{tot}(0) = \rho_S(0) \otimes \rho_B(0)$, the exact dynamics of $S$ can indeed be described by a HEOM which takes the form~\cite{Tanimura89, Ishizaki2005, Ishizaki2009, Yoshitaka2020, Lambert2020Qutip}
\begin{equation} \label{HEOM}
\begin{aligned}
    \frac{d\rho^{(\vec{n}, \vec{m})}}{dt} &= -i[H_S, \rho^{(\vec{n}, \vec{m})}] - (\vec{w}^*\bcdot\vec{n} + \vec{w}\bcdot\vec{m}) \rho^{(\vec{n}, \vec{m})} \\
    &+ \sum_{l = 1}^{N_E}\sum_{j = 1}^{M_{l}}\left\{ G_{lj} n_{lj} L_l \rho^{(\vec{n}-\vec{e}_{lj}, \vec{m})} + G_{lj}^* m_{lj} \rho^{(\vec{n}, \vec{m}-\vec{e}_{lj})}L_{l}^\dagger \right. \\
    &\qquad \quad  + \left.[\rho^{(\vec{n}+\vec{e}_{lj}, \vec{m})}, L_l^\dagger] + [L_l, \rho^{(\vec{n}, \vec{m}+\vec{e}_{lj})}]\right\},
\end{aligned}
\end{equation}
with $\vec{n} = (n_{lj})$ and $\vec{m}= (m_{lj})$ vectorized sets of multi-indices in $\mathbb{N}^M$ with $M = \sum_l M_l$ the total number of pseudomodes~\cite{footnote_multiindices}, $\vec{w} = (\kappa_{lj} + i\, \omega_{lj}) \in \mathbb{C}^M$, $\vec{e}_{lj} = (\delta_{jj'}\delta_{ll'})$ unit vectors, and $\vec{a}\bcdot \vec{b} = \sum_{lj} a_{lj}^* b_{lj}$ the inner product on $\mathbb{C}^M$. Equation~(\ref{HEOM}) is a linear set of coupled master equations for the set of operators $\{\rho^{(\vec{\vec{n}}, \vec{m})}\}$, where $\rho^{(\vec{0}, \vec{0})} \equiv \rho_S$ corresponds to the physical density operator of the system $S$ with which all the expectation values of system observables are computed. The operators $\rho^{(\vec{n}, \vec{m})}$ for $(\vec{n}, \vec{m}) \neq (\vec{0},\vec{0})$ also act on the system space and are called auxiliary states in the literature (even if they are generally not Hermitian and of unit trace). From these auxiliary states, which encode the build-up of correlations between the system and the environment, bath correlation functions can be obtained, as shown in Sec.~II.C. For $t = 0$, we have $\rho^{(\vec{0}, \vec{0})}(0) = \rho_S(0)$ and $\rho^{(\vec{n}, \vec{m})}(0) = 0$ for $(\vec{n}, \vec{m}) \neq (\vec{0},\vec{0})$. As time evolves, the auxiliary states become progressively populated and are likely to influence the dynamics of the physical state at later times, acting in this way as a memory kernel for the system.
Although the hierarchy is formally infinite, it can be truncated at large hierarchy depth indices $\vec{n}$ and $\vec{m}$, typically because some auxiliary states stay negligibly populated at all times.
In practice, the stronger the non-Markovianity, the larger the number of auxiliary states we need to retain to obtain convergence of the results. Here, we choose the triangular truncation $\rho^{(\vec{n}, \vec{m})} = 0 \; \forall \:\vec{n},\vec{m}: \sum_{l,j} (n_{lj} + m_{lj}) > k_\mathrm{max}$, where $k_\mathrm{max}$ is the truncation order, giving a total of
\begin{equation}\label{K}
    K = \frac{(2M+k_{\mathrm{max}})!}{(2M)!\, k_{\mathrm{max}}!}
\end{equation}
auxiliary states~\cite{Link2022}.

Defining $\rho$ as the vector that gathers all the operators $\rho^{(\vec n, \vec m)}$, Eq.~(\ref{HEOM}) is formally equivalent to 
\begin{equation}
    \frac{d\rho}{dt} = \mathcal{L}_\mathrm{HEOM}[\rho],
\end{equation}
where we introduced the HEOM Liouvillian $\mathcal{L}_\mathrm{HEOM}$, i.e., the superoperator that generates the dynamics of the physical density operator $\rho_S = \rho^{(\vec0, \vec0)}$ and all the auxiliary operators $\rho^{(\vec n, \vec m )}$. 
Moreover, the Choi-Jamiolkowsi isomorphism directly offers a matrix representation of the HEOM generator 
\begin{equation}
    \frac{d\dket{\rho}}{dt} = \mathcal{L}_{\mathrm{HEOM}}(k_{\mathrm{max}})\,\dket{\rho},
    \label{def_LHEOM}
\end{equation}
with $\dket{\rho}$ a stacked vector containing all the vectorized versions $\dket{\rho^{(\vec n, \vec m)}}$ of the matrices $\rho^{(\vec{n}, \vec{m})}$, while the notation $\mathcal{L}_\mathrm{HEOM}(k_\mathrm{max})$ highlights the $k_\mathrm{max}$ dependency of $\mathcal{L}_\mathrm{HEOM}$. Note that we keep the same notations for the superoperator $\mathcal{L}_\mathrm{HEOM}$ and its matrix representation; the context should leave no room for ambiguity. Explicit forms of the HEOM Liouvillian for different $k_\mathrm{max}$ are given in Appendix~\ref{appExplicitform}. 

\subsection{Spectral properties of the HEOM Liouvillian}
\label{Spectral_Properties}
The HEOM Liouvillian $\mathcal{L}_\mathrm{HEOM}(k_\mathrm{max})$ is a linear and, in general, non-Hermitian superoperator. In the following, we assume that it is diagonalizable and denote its eigenvalues by $\lambda_i$ and the corresponding right eigenvectors by $\dket{\rho_i}$, that is, 
\begin{equation}
    \mathcal{L}_\mathrm{HEOM}(k_\mathrm{max}) \dket{\rho_i} = \lambda_i \dket{\rho_i}.
\end{equation}
For a truncation order $k_{\mathrm{max}}$, its dimension is $
D = K\,\mathrm{dim(\mathcal{H}_S})^2$ with $K$ given in Eq.~\eqref{K}. For any value of $k_\mathrm{max}\in \mathbb{N}$, the following properties for $\mathcal{L}_\mathrm{HEOM}(k_\mathrm{max})$ hold: 
\begin{enumerate}[label = (\roman*)] 
    \item its spectrum is symmetric with respect to the real axis;
    \item it preserves the trace of the physical state $\rho^{(\vec{0},\vec{0})}$;
    \item the $0$ eigenvalue is always in its spectrum, which guarantees the existence of a stationary state. 
\end{enumerate}
Furthermore, owing to the exactness of the generator $\mathcal{L}_\mathrm{HEOM}(k_\mathrm{max})$ in the limit $k_\mathrm{max} \rightarrow + \infty$, we have that
\begin{enumerate}[label = (\roman*)]
    \setcounter{enumi}{3}
    \item all the eigenvalues have a negative real part;
    \item $\mathrm{Tr}[\mathds{1}^{(\vec{0}, \vec{0})}\rho_i] = 0$ with $\mathds{1}^{(\vec{0}, \vec{0})}$ the projector onto the physical state space and with $\rho_i$ a right eigenoperator associated with the eigenvalue $\lambda_i$, with $\mathrm{Re}[\lambda_i] \neq 0$.
\end{enumerate}
All these spectral properties are proved in Appendix~\ref{Spectral_Proofs}. As in~\cite{Minganti2018Spectral, Kessler2012Dissipative}, we order the eigenvalues of $\mathcal{L}_{\mathrm{HEOM}}$ by their real part  so that
$|\mathrm{Re}[\lambda_0]| < |\mathrm{Re}[\lambda_1]| < \dots < |\mathrm{Re}[\lambda_D]|$, where $\lambda_0=0$.

\subsection{Advantage of $\mathcal{L}_{\mathrm{HEOM}}$ over Markovian embedding}

\begin{figure*}
    \centering
\includegraphics[width=0.975\textwidth]{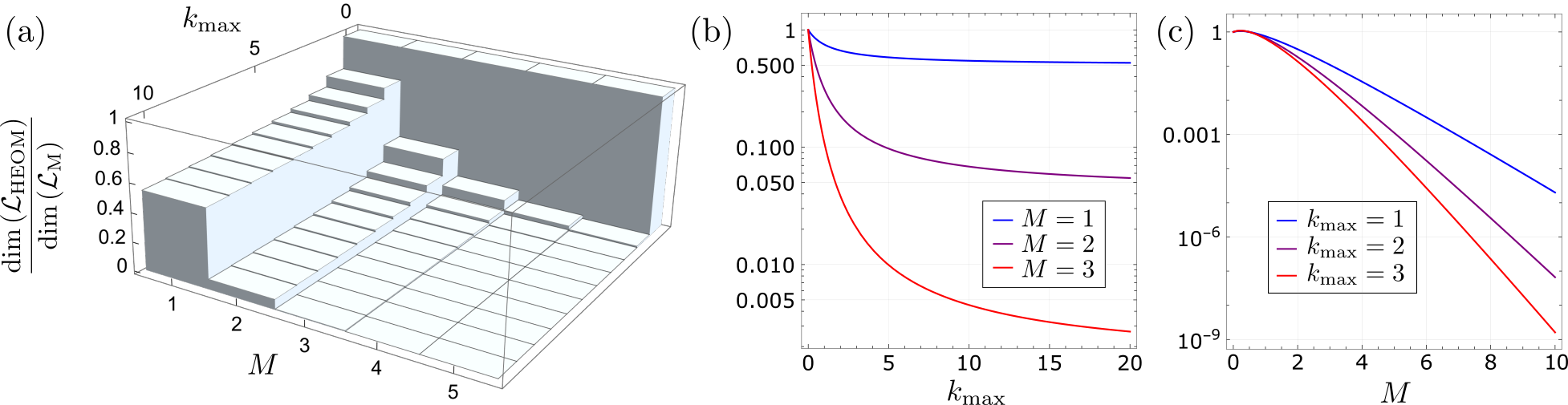}
    \caption{Overall comparison between the dimensions of $\mathcal{L}_\mathrm{HEOM}$ [Eq.~(\ref{dimLHEOM})] and $\mathcal{L}_\mathrm{M}$ [Eq.~(\ref{dimLM})]: ratios $\mathrm{dim}\left(\mathcal{L}_\mathrm{HEOM}\right)/\mathrm{dim}\left(\mathcal{L}_\mathrm{M} \right)$ as a function $k_\mathrm{max}$ and $M$ (a), as a function of $k_\mathrm{max}$ for $M = 1,2$ and $3$ (b), and as a function of $M$ for $k_\mathrm{max} = 1,2$ and $3$ (c). Since $\mathrm{dim}\left(\mathcal{L}_\mathrm{HEOM}\right)/\mathrm{dim}\left(\mathcal{L}_\mathrm{M} \right) < 1$, this means that we need less computational memory to store $\mathcal{L}_\mathrm{HEOM}$ than $\mathcal{L}_\mathrm{M}$.}
    \label{comparison1}
\end{figure*}
Before moving on to the spectral theory of non-Markovian DPTs, it is worth comparing the numerical effort of the HEOM method [Eq.~(\ref{def_LHEOM})] with the natural Markovian embedding that includes the pseudomodes [Eq.~(\ref{MEMArkov})]. This provides an overall idea of what kind of computational advantage can be expected when using $\mathcal{L}_\mathrm{HEOM}$ instead of $\mathcal{L}_\mathrm{M}$. A more detailed comparison can be found in Appendix~\ref{Numerical_Efficiency}.

The dimension of the matrix representing $\mathcal{L}_\mathrm{HEOM}$ is given by
\begin{equation}\label{dimLHEOM}
\mathrm{dim}\left(\mathcal{L}_\mathrm{HEOM} \right) = \frac{(2M+k_{\mathrm{max}})!}{(2M)! k_{\mathrm{max}}!}\mathrm{dim(\mathcal{H}_S})^2
\end{equation}
and depends on the size of the system Hilbert space $\mathcal{H}_S$, the truncation order $k_\mathrm{max}$ and the number of pseudomodes $M$ which determines the number of auxiliary matrices in the hierarchy.

To compare with the dimension of the matrix representing the Liouvillian of the enlarged Markovian system $\mathcal{L}_\mathrm{M}$, we need to introduce a cutoff $N_c$ for the pseudomode Fock spaces $\left\{|n_i\rangle\right\}$ ($n_i = 0,1,\dotsc,\infty$ and $i = 1,2, \dotsc,M$), which are in principle of infinite dimension. We choose here $N_c = k_\mathrm{max}$, motivated by the fact that the pseudomode correlation functions are related to the traces of the auxiliary matrices according to (for $M = 1$)~\cite{Link2022}
\begin{equation}
    \langle a^n (a^\dagger)^m \rangle (t) = \frac{\mathrm{Tr}\left[ \rho^{(n,m)}(t)\right]}{(i G)^n (-i G)^m},
\end{equation}
which means that if we truncate the hierarchy at $k_\mathrm{max}$, we must at least truncate the pseudomode Fock space at $N_c = k_\mathrm{max}$ to be able to compute the same correlations. The dimension of $\mathcal{L}_\mathrm{M}$ should thus be
\begin{equation}\label{dimLM}
\mathrm{dim}\left(\mathcal{L}_\mathrm{M} \right) = \mathrm{dim(\mathcal{H}_S})^2 (k_\mathrm{max}+1)^M.
\end{equation}
The ratio $\mathrm{dim}\left(\mathcal{L}_\mathrm{HEOM} \right)/\mathrm{dim}\left(\mathcal{L}_\mathrm{M} \right)$ is plotted as a function of $k_\mathrm{max}$ and $M$ in Fig.~\ref{comparison1}. We can see that the advantage in terms of dimension can be significant, especially for large numbers of pseudomodes.

\section{Non-Markovian spectral theory of DPTs}
\subsection{Definition of DPTs}
As for the Markovian case~\cite{Minganti2018Spectral}, we define a phase transition as the emergence of a nonanalytic behavior of the steady state as some control parameter $g$ is varied. More precisely, for any system described by Eq.~\eqref{def_LHEOM} which supports a valid thermodynamic limit $N \to \infty$ and possesses a unique steady state $\rho_{ss}$ for all finite $N$, if 
 \begin{equation}\label{defDPT}
     \lim_{g\to g_c} \left|\frac{\partial^p}{\partial g^p} \lim_{N\to +\infty} \langle O \rangle_{ss}\right| = +\infty,
 \end{equation}
where $O$ is a $g$-independent system observable and $\langle O \rangle_{ss} = \mathrm{Tr}[O\rho_{ss}^{(\vec{0},\vec{0})}]$, we say that the system undergoes a phase transition of order $p$.  We call $g_c$ the critical point or critical value of the control parameter. We stress that we \textit{always} assume that the steady state is unique at finite $N$.\\ 

\subsection{Weak symmetries and spontaneous symmetry breaking}\label{sec:symSSB}
Here we define a \textit{weak symmetry} of $\mathcal{L}_{\mathrm{HEOM}}$ as a unitary superoperator $\mathcal{U}$ which commutes with $\mathcal{L}_{\mathrm{HEOM}}$, i.e.\
\begin{equation}
    [\mathcal{L}_{\mathrm{HEOM}}, \mathcal{U}] =0.
    \label{def_U}
\end{equation}
In the eigenvector basis of $\mathcal{U}$, $\mathcal{L}_\mathrm{HEOM}$ is block-diagonal and we can write
\begin{equation}\label{blockLHEOM}
    \mathcal{L}_{\mathrm{HEOM}} = \bigoplus_ {u_k} \mathcal{L}_{u_k},
\end{equation}
where each block $\mathcal{L}_{u_k}$ is associated with distinct eigenvalues $u_k$ of $\mathcal{U}$ where $k \in \{0, 1,\dotsc\}$. The \textit{symmetry sector} $L_{u_k}$ is defined as the subspace spanned by the eigenvectors of $\mathcal{U}$ associated with the eigenvalue $u_k$. As for the Markovian case~\cite{Minganti2018Spectral}, the steady state always belong to the symmetry sector $L_{u_0 = 1}$. Indeed, by definition of the steady state, we have $\mathcal{L}_{\mathrm{HEOM}} \dket{\rho_{ss}} = 0$, from which we get $
\mathcal{U}\mathcal{L}_{\mathrm{HEOM}}\dket{\rho_{ss}} = \mathcal{L}_{\mathrm{HEOM}}\mathcal{U}\dket{\rho_{ss}} =0$ by definition of weak symmetry [Eq.~\eqref{def_U}]. Since $\dket{\rho_{ss}}$ is unique by hypothesis, we must have $\mathcal{U}\dket{\rho_{ss}} = \dket{\rho_{ss}}$, which shows that $\dket{\rho_{ss}}\in L_{u_0=1}$.\\

Spontaneous symmetry breaking (SSB) is defined as the emergence of a zero eigenvalue in each symmetry sector $k \neq 0$ in the limit $N \to \infty$. Specifically, if a certain weak symmetry $\mathcal{U}$ gives rise to $n+1$ symmetry sectors $L_{u_k}$ with $k\in\{0, 1, ..., n\}$, then a SSB occurs when there exists in each symmetry sector an eigenvalue that converges to zero in the thermodynamic limit, i.e., 
\begin{equation}
    \lim_{N \rightarrow + \infty}\lambda_0^{(k)} \rightarrow 0, \quad k\in\{1, ..., n\},
\end{equation}
where we sorted the eigenvalues in each symmetry sector $L_{u_k}$ as follows 
\begin{equation}
     |\mathrm{Re}[\lambda_0^{(k)}]| < |\mathrm{Re}[\lambda_1^{(k)}]| < \dots
\end{equation}
This means that the independent hierarchies associated with each block $k$, forced to be independent by the very existence of the weak symmetry, merge in the thermodynamic limit $N\rightarrow +\infty$. Consequently, steady states that explicitly break the symmetry emerge when $N \rightarrow +\infty$.

\subsection{Spectral theory of non-Markovian DPTs}\label{sec::Spectral theory}
In this section, we generalize the Markovian theory of DPTs~\cite{Minganti2018Spectral} to the non-Markovian regime. We discuss how the properties of the HEOM Liouvillian are connected to $1^{\mathrm{st}}$-order DPTs and to $2^{\mathrm{nd}}$-order DPTs associated with spontaneous symmetry breaking (SSB).

\subsubsection{$1^{\mathrm{st}}$-order DPTs}

Since $1^{\mathrm{st}}$-order DPTs are independent of symmetries (and thus of SSB), we consider below no particular symmetry and simply label the eigenvectors and eigenvalues of $\mathcal{L}_\mathrm{HEOM}$ as $\rho_i$ and $\lambda_i$. In the case of $1^{\mathrm{st}}$-order DPTs emerging in systems with symmetries, the following results must be understood as related to the Liouvillian block associated with the symmetry sector containing the steady state, i.e., the block $\mathcal{L}_{u_0 = 1}$ of the decomposition \eqref{blockLHEOM}, so that $\rho_i$ and $\lambda_i$ below must be understood simply as $\rho_i^{(0)}$ and $\lambda_i^{(0)}$. This is exemplified in the companion paper~\cite{Debecker2024PRL}, where we show a first-order DPT occurring concurrently with an SSB.

Due to the spectral properties (i)-(v) of~Section~\ref{Spectral_Properties} and their similarity to the Markovian case~\cite{Minganti2018Spectral}, it can be shown that a $1^{\mathrm{st}}$-order DPT can occur if and only if the HEOM Liouvillian gap $\mathrm{Re}[\lambda_1]$ vanishes at the critical point in the thermodynamic limit. Moreover, $\mathrm{Im}[\lambda_1]$ must vanish in a finite domain around the critical point. We prove below these statements in the limit $k_\mathrm{max} \rightarrow + \infty$. The proofs closely follow Ref.~\cite{Minganti2018Spectral} which itself relies on the results of Kato~\cite{Kato} and ~\cite{Macieszczak2016T}. We start with a definition of a key superoperator.\\

Let $\rho_i$ be a right eigenoperator of $\mathcal{L}_\mathrm{HEOM}$. We define the superoperator $\mathcal{P}^{(\vec0, \vec0)}$ through
  \begin{equation}
     \mathcal{P}^{(\vec0, \vec0)} \rho_i = \rho_i^{(\vec0, \vec0)} \quad\Leftrightarrow\quad \dket{\mathcal{P}^{(\vec0, \vec0)} \rho_i} = \dket{\rho_i^{(\vec0, \vec0)}},
  \end{equation}
  i.e., $\mathcal{P}^{(\vec0, \vec0)}$ only selects the component of $\rho_i$ in the physical sector $(\vec0, \vec0)$, that is the operator $\rho_i^{(\vec0, \vec0)}$ acting on a space of dimension $\mathrm{dim}\left(\mathcal{H}_\mathrm{S} \right)$. Note that $\mathcal{P}^{(\vec0, \vec0)}$ is different from $\mathds{1}^{(\vec0, \vec0)}$, as $\dket{\mathds{1}^{(\vec0, \vec0)} \rho_i}$ corresponds to the stacked vector $(\dket{\rho_i^{(\vec0, \vec0)}}, \dket{0}, \dket{0}, \dotsc)^T$. We are now in a position to prove that a first order DPT implies the vanishing of the gap exactly at the critical point.

\begin{theorem}\label{prop1}
If a physical system undergoes a $1^{\mathrm{st}}$-order DPT in a well-defined thermodynamic limit $N\rightarrow +\infty$ at the critical point $g=g_c$ separating two unique phases, and if $\lim_{N\rightarrow +\infty}\mathcal{L}_\mathrm{HEOM}(g, N)$ is continuous with respect to $g$, then~$\lim_{N \rightarrow +\infty}\lambda_1(g=g_c, N) = 0$.
\end{theorem}
\begin{proof}
Let us assume that a system undergoes a $1^{\mathrm{st}}$-order DPT in a well-defined thermodynamic limit $N\rightarrow +\infty$ when a parameter $g$ is varied. By definition, the steady state $\rho^{(\vec{0}, \vec{0})} (t \rightarrow \infty) \equiv \rho_{ss}^{(\vec{0}, \vec{0})} $ must change discontinuously, which implies that there exists a critical point $g_c$ such that
\begin{equation}
    \rho^{(\vec 0, \vec 0)}_- \neq  \rho^{(\vec 0, \vec 0)}_+,
    \label{Defrho+rho-}
\end{equation}
where $\rho^{(\vec 0, \vec 0)}_-$ and $\rho^{(\vec 0, \vec 0)}_{+}$ are the states (phases) of the system right before and after the transition point, i.e.,
\begin{equation}
 \begin{split}
    \rho^{(\vec 0, \vec 0)}_- &\equiv \lim_{g \rightarrow g^-_c} \lim_{N\rightarrow +\infty} \rho^{(\vec 0, \vec 0)}_{ss}(g, N), \\
     \rho^{(\vec 0, \vec 0)}_+ &\equiv \lim_{g \rightarrow g^+_c} \lim_{N\rightarrow +\infty} \rho_{ss}^{(\vec 0, \vec 0)}(g, N) 
 \end{split}
\end{equation}
which are unique by hypothesis. This implies that there exists $\rho_\pm$ such that 
\begin{equation}
\left(\lim_{g\rightarrow g_c^\pm}\lim_{N\rightarrow +\infty} \mathcal{L}_\mathrm{HEOM}(g, N)\right) [\rho_\pm] =0 ,
\end{equation}
with $\mathcal{P}^{(\vec0, \vec0)}\rho_\pm = \rho_\pm^{(\vec0, \vec0)}$. The continuity of the HEOM generator in the thermodynamic limit then gives 
\begin{equation}
   \left( \lim_{g\rightarrow g_c}\lim_{N\rightarrow +\infty} \mathcal{L}_\mathrm{HEOM}(g, N)\right) [\rho_\pm] =0. 
\end{equation}
Exactly at $g=g_c$, we then have two eigenoperators that belong to the null space of $\mathcal{L}_\mathrm{HEOM}$, while there is a unique steady state for $g \neq g_c$. Consequently, we must have 
\begin{equation}
    \lim_{N \rightarrow +\infty}\lambda_1(g = g_c, N) = 0 = \lambda_0, \quad \lim_{N \rightarrow +\infty}\lambda_1(g \neq g_c, N) \neq 0,
    \label{SM_lambda1geqgc}
\end{equation}
which concludes the proof.
\end{proof}
 Note that even if $\mathcal{L}_\mathrm{HEOM}(g)$ is continuous, there is no guarantee that $\rho_1(g)$ is continuous. Indeed, the coalescence of eigenvalues for instance may induce non-continuous eigenvectors~\cite{Kato}. Nevertheless, we will get rid of these difficulties by assuming that $\rho_1^{(\vec 0, \vec 0)} \equiv \mathcal{P}^{(\vec{0}, \vec{0})}\rho_1$ is continuous, as is done in the Markovian case~\cite{Minganti2018Spectral}. We can then elaborate on the form of the steady state at the critical point as a function of the right eigenvectors of the Liouvillian associated with $\lambda_0$ and $\lambda_1$, that is $\rho_0$ and $\rho_1$. In the following, unless otherwise stated we always assume the thermodynamic limit and thus remove the $N$ dependency.
\begin{theorem}\label{prop2}
Under the same assumptions as in Proposition \ref{prop1}, if $\rho_+$ and $\rho_-$ span the null space at $g=g_c$, then $\mathrm{Im}[\lambda_1] = 0$ holds in a finite neighborhood of $g=g_c$. Moreover, if $\lim_{N\rightarrow +\infty}\rho_1^{(\vec0, \vec0)}(g, N)$ is continuous with respect to $g$ and orthogonal to the steady state, then 
\begin{equation}
    \rho_{ss}^{(\vec0, \vec0)}(g=g_c) = \frac{\rho_+^{(\vec0, \vec0)} + \rho_-^{(\vec0, \vec0)}}{2},
    \label{rhossAtgc1stOrder}
\end{equation}
in the thermodynamic limit, \textit{i.e.,} the steady state is an equal mixture of the two phases $\rho^{\pm}$.
\end{theorem}
\begin{proof}
We must have $\mathrm{Im}[\lambda_1] = 0$ in a finite neighbourhood of $g_c$ as, if it were not the case, then we would have three zero eigenvalues at $g=g_c$ :  $\lambda_1(g=g_c) = \lambda_1^*(g=g_c) = \lambda_0 = 0$ due to the spectral property (i). This would imply that $\rho^+$ and $\rho^-$ do not span the null space at $g=g_c$~\footnote{We recall that we always assume the diagonalizability of the HEOM Liouvillian.}, which contradicts our assumptions. 

To prove Eq.~\eqref{rhossAtgc1stOrder}, we first note that from the property (v) of $\mathcal{L}_{\mathrm{HEOM}}$, we have $\mathrm{Tr}[\mathds{1}^{(\vec 0, \vec 0)}\rho_1 (g \neq g_c)] = 0$ since $\mathrm{Re}[\lambda_1(g\neq g_c)] \neq 0$. By hypothesis, $\rho_1^{(\vec 0, \vec 0)} \equiv \mathcal{P}^{(\vec{0}, \vec{0})}\rho_1$ is continuous. This means that $\rho_1^{(\vec 0, \vec 0)}(g)$ does not change abruptly as a function $g$ but instead evolves continuously and makes a small excursion into the null space of $\mathcal{L}_\mathrm{HEOM}$ exactly at $g=g_c$. Therefore, we extend the zero trace property of $\rho_1^{(\vec 0, \vec 0)}$ at the critical point. Since the system null space at $g=g_c$ is spanned by $\rho^{(\vec 0, \vec 0)}_{\pm}$, we must have 
\begin{equation}
    \rho_1^{(\vec 0, \vec 0)}(g=g_c) \propto \rho^{(\vec 0, \vec 0)}_+ - \rho^{(\vec 0, \vec 0)}_-.
    \label{rho1_pm}
\end{equation}
Furthermore, since $\rho_\pm^{(\vec0, \vec0)}$ span the null space at $g=g_c$, the steady state at the critical point must read 
\begin{equation}
    \rho_{ss}(g=g_c) = c\rho_+^{(\vec0, \vec0)} + (1-c) \rho_-^{(\vec0, \vec0)},
\end{equation}
with $c \in [0, 1]$. The precise value of $c$ cannot be determined through the behavior of $\rho_{ss}(g\neq g_c)$. If, however, we impose orthogonality between $\rho_{ss}(g=g_c)$ and  $\rho_1^{(\vec0, \vec0)} (g=g_c)$, we obtain 
\begin{equation}
    \rho_{ss}^{(\vec0, \vec0)}(g=g_c) = \frac{\rho_+^{(\vec0, \vec0)} + \rho_-^{(\vec0, \vec0)}}{2},
\end{equation}
in virtue of the orthogonality between $\rho_\pm^{(\vec0, \vec0)}$, that we prove below.

For $\lambda_1(g=g_c) = 0 \in \mathbb{R}$, $\rho_1^{(\vec0, \vec0)}$ can be assumed Hermitian without loss of generality, as a direct consequence of property (i) of $\mathcal{L}_{\mathrm{HEOM}}$. Consequently, $\rho_1^{(\vec0, \vec0)}$ can be diagonalized and one can construct $\rho_{\pm}$, satisfying Eq.~\eqref{rho1_pm}, simply by gathering all positive eigenvalues in $\rho^{(\vec 0, \vec 0)}_{+}$ and all negative eigenvalues in $-\rho^{(\vec 0, \vec 0)}_{-}$ and then normalizing the trace of $\rho^{(\vec 0, \vec 0)}_{\pm}$ to one. By construction, $\rho^{(\vec 0, \vec 0)}_{\pm}$ are orthogonal.
\end{proof}
We stress that the decomposition~\eqref{rhossAtgc1stOrder} is based on strong assumptions, \textit{e.g.,} the orthogonality of $\rho_1^{(\vec0, \vec0)}$ and $\rho_{ss}$ at the critical point, which may not be met. Indeed, in Ref.~\cite{Huber2020Nonequilibrium}, it has been shown that, in the Markovian regime, there exists $1^{\mathrm{st}}$-order DPTs without phase coexistence at the critical point, which contradicts Eq.~\eqref{rhossAtgc1stOrder}. Nevertheless, as shown in the next section, we find an excellent agreement with the present theory and the two models explored in the companion paper~{\cite{Debecker2024PRL}. We close this section dedicated to first order DPTs by proving that the closing of the gap at the critical point implies the existence of a first order DPT.

\begin{theorem}\label{prop3}
If $\mathcal{\lambda}_1(g)$ vanishes only at $g=g_c$ with $\mathrm{Im}[\lambda_1(g)] = 0$ in a finite domain around $g=g_c$, $\rho_1^{(\vec0, \vec0)} (g)$ is continuous and if the null space is spanned at $g=g_c$ by two linearly independent eigenoperators, then there is a $1^{\mathrm{st}}$-order DPT occurring at $g=g_c$.
\end{theorem}
\begin{proof}
    We proceed by contradiction by assuming that $\lim_{g\rightarrow g_c} \lambda_1(g) = 0$ and that there is no $1^{\mathrm{st}}$-order DPT. Equivalently, we have that for any observable $O$ of the system $\mathrm{Tr}[O \rho_\mathrm{ss}^{(\vec 0, \vec 0)}](g)$ is continuous at $g=g_c$, which implies that $\rho_{ss}^{(\vec 0, \vec 0)}(g)$ is also continuous at $g=g_c$. However, at $g=g_c$, $\lambda_1 =0$ and as in Proposition \ref{prop2}, we may then extend the zero-trace condition by setting
\begin{equation}
    \rho_{1}^{(\vec 0, \vec 0)} (g=g_c) \propto \rho^{(\vec 0, \vec 0)}_{1+} - \rho^{(\vec 0, \vec 0)}_{1-},
\end{equation}
where $\rho^{(\vec 0, \vec 0)}_{1\pm}$ can be found by diagonalizing $\rho_1^{(\vec0, \vec0)}$ [which is always possible since $\rho_1^{(\vec0, \vec0)}$ is hermitian in the domain in which $\mathrm{Im}[\lambda_1] = 0$] and gathering again all positive eigenvalues in $\rho^{(\vec 0, \vec 0)}_{1+}$ and all negative eigenvalues in $-\rho^{(\vec 0, \vec 0)}_{1-}$ and then normalizing the trace of $\rho^{(\vec 0, \vec 0)}_{1\pm}$ to one. By construction,  $\rho^{(\vec 0, \vec 0)}_{1\pm}$ are density matrices such that $\dbraket{\rho^{(\vec 0, \vec 0)}_{1+}|\rho^{(\vec 0, \vec 0)}_{1-}} \equiv \mathrm{Tr}[\rho^{(\vec 0, \vec 0)\dagger}_{1+} \rho^{(\vec 0, \vec 0)}_{1-}] = 0$. At the critical point, the steady state can then be written as 
\begin{equation}
  \rho_{ss}^{(\vec0, \vec0)} (g=g_c) = c \rho^{(\vec 0, \vec 0)}_{1+} + (1-c) \rho^{(\vec 0, \vec 0)}_{1-},
\end{equation}
where $c \in [0, 1]$. Now, for all $g \neq g_c$ and $g\in [g_c-\epsilon, g_c + \epsilon]$ ($\epsilon > 0$), the gap is not closed. Therefore, 
\begin{equation}
 \begin{split}
   \mathcal{P}^{(\vec0, \vec0)}\lim_{g\rightarrow g_c^\pm}\lim_{t\rightarrow +\infty} e^{\mathcal{L}_\mathrm{HEOM} (g) t} \rho_{1\pm}(g) &= \rho_0^{(\vec0, \vec0)}(g_c^\pm) \\ 
   &= \rho^{(\vec0, \vec0)}_{ss}(g_c^\pm),
   \end{split}
\end{equation}
By hypothesis, however, $\rho_{ss}^{(\vec0, \vec0)}$ is continuous at $g=g_c$, which means 
\begin{equation}
 \begin{split}
     \mathcal{P}^{(\vec0, \vec0)}\lim_{g\rightarrow g_c^\pm}\lim_{t\rightarrow +\infty} e^{\mathcal{L}_\mathrm{HEOM} (g) t} \rho_{1\pm}(g) &= \rho_{ss}^{(\vec0, \vec0)} (g=g_c) \\
     &= c \rho^{(\vec 0, \vec 0)}_{1+} + (1-c) \rho^{(\vec 0, \vec 0)}_{1-},
 \end{split}
\end{equation}
or
\begin{equation}
    \rho^{(\vec 0, \vec 0)}_{1\pm} (g=g_c) = c \rho^{(\vec 0, \vec 0)}_{1+}(g=g_c) + (1-c) \rho^{(\vec 0, \vec 0)}_{1-}(g=g_c),
    \label{rho_1^pm}
\end{equation}
since $\rho_{1\pm}(g)$ belongs to the null space at $g=g_c$. From Eq.~\eqref{rho_1^pm} we infer $\rho^{(\vec 0, \vec 0)}_{1+}(g=g_c) = \rho^{(\vec 0, \vec 0)}_{1-}(g=g_c)$, hence the contradiction with $\dbraket{\rho^{(\vec 0, \vec 0)}_{1+}|\rho^{(\vec 0, \vec 0)}_{1-}} \equiv \mathrm{Tr}[\rho^{(\vec 0, \vec 0)\dagger}_{1+} \rho^{(\vec 0, \vec 0)}_{1-}] = 0$ or even with the very existence of $\rho_1^{(\vec0, \vec0)}$.
\end{proof}

\subsubsection{$2^{\mathrm{nd}}$-order DPTs with SSB}

Let us now discuss the consequences of the properties of the HEOM Liouvillian on $2^{\mathrm{nd}}$-order DPTs associated with SSB. For clarity of exposition, we only consider DPTs associated with a $\mathbb{Z}_2$-SSB, but the generalization to more general symmetries is straightforward.

Let $\mathcal{L}_\mathrm{HEOM} (g)$ be the HEOM generator that captures a $2^{\mathrm{nd}}$-order DPT associated with the spontaneous breaking of a $\mathbb{Z}_2$ symmetry for $g \geq g_c$, $g_c$ being the critical point. By definition of a weak symmetry, there exists a superoperator $\mathcal{U}_2$ such that 
\begin{equation}
    [\mathcal{L}_\mathrm{HEOM}, \mathcal{U}_2] = 0.
\end{equation}
As discussed in Section~\ref{sec:symSSB}, the very existence of the operator $\mathcal{U}_2$ constrains the HEOM generator to adopt a block-diagonal structure when written in the eigenbasis of $\mathcal{U}_2$, namely \begin{equation}
    \mathcal{L}_\mathrm{HEOM} = \bigoplus_{u_k = \pm 1} \mathcal{L}_{u_k},
\end{equation}
where $\mathcal{L}_{u_k = \pm1}$ are the two blocks associated with the two eigenvalues of $\mathcal{U}_2$, namely $+1$ ($k = 0$) and $-1$ ($k = 1$). By hypothesis, a $2^{\mathrm{nd}}$-order DPT with $\mathbb{Z}_2$-SSB occurs for $g \geq g_c$ and $N \rightarrow + \infty$, i.e., the two blocks $\mathcal{L}_{+1}$ and $\mathcal{L}_{-1}$ get coupled in the thermodynamic limit: $\lambda_0^{(k=0)}(g \geq g_c)= 0 = \lambda_0^{(k=1)}$. The associated eigenvectors, denoted by $\rho_0^{(k)}$, are then orthogonal since we have
\begin{equation}
 \begin{split}
    \dbraket{\rho_0^{(0)} | \rho_0^{(1)}} &= \dbraket{ \mathcal{U}_2\rho_0^{(0)} | \rho_0^{(1)}} \\
    &= \dbraket{\rho_0^{(0)} | \mathcal{U}_2 \rho_0^{(1)}} \\
    &= - \dbraket{\rho_0^{(0)} | \rho_0^{(1)}}
 \end{split}
\end{equation}
as $\mathcal{U}_2$ is Hermitian, and thus $\dbraket{\rho_0^{(0)} | \rho_0^{(1)}} = 0$. Now, if we define  
\begin{equation}
   \rho_\pm \propto \rho_0^{(0)} \pm \rho_0^{(1)},
   \label{rhopm_SSB}
\end{equation}
it is clear that $\rho_\pm$ belong to the kernel of $\mathcal{L}_\mathrm{HEOM}(g\geq g_c)$ in the thermodynamic limit. Note, however, that $\rho_\pm$ are not eigenvectors of $\mathcal{U}_2$ as $\mathcal{U}_2 \rho_\pm \propto \rho_\mp$. Moreover, the projections $\rho^{(\vec0, \vec0)}_{\pm} = \mathcal{P}^{(\vec0, \vec0)} \rho_\pm$ allow us to interpret $\rho^{(\vec0, \vec0)}_{\pm}$ as steady states of the system that explicitly break the symmetry. Equation~(\ref{rhopm_SSB}) can be inverted, so that 
\begin{equation}
 \begin{split}
    \rho_0^{(0)} &\propto \rho_+ + \rho_-, \\
    \rho_0^{(1)} &\propto \rho_+ - \rho_-.
 \end{split}
\end{equation}
In particular, if we apply $\mathcal{P}^{(\vec0, \vec0)}$ to both sides of the previous relations, we obtain 
\begin{equation}
 \begin{split}
     \rho_0^{(0)(\vec0, \vec0)} &\propto \rho^{(\vec0, \vec0)}_+ + \rho^{(\vec0, \vec0)}_-, \\
    \rho_0^{(1)(\vec0, \vec0)} &\propto \rho^{(\vec0, \vec0)}_+ - \rho^{(\vec0, \vec0)}_-,
 \end{split}
\end{equation}
with $\mathcal{P}^{(\vec0, \vec0)}\rho_0^{(k)} = \rho_0^{(k)(\vec0, \vec0)}$ ($k = 0, 1$), which is exactly what the Markovian theory predicts~\cite{Minganti2018Spectral}. For finite $N$, the steady state is unique, \textit{i.e.} $\rho_{ss} \propto \rho_0^{(0)(\vec0, \vec0)}$ and therefore
\begin{equation}
    \rho_{ss} (g \geq g_c, N) \approx \frac{\rho_+^{(\vec0, \vec0)} (g \geq g_c, N) + \rho_-^{(\vec0, \vec0)}(g \geq g_c, N)}{2}.
    \label{rhoApprox2ndOrder}
\end{equation}

To conclude this section, a summary of the behavior of the key eigenvalues of the HEOM Liouvillian for the different cases is depicted in Fig.~\ref{SM_FigSpectrum}.

\begin{figure*}
    \centering
    \includegraphics[width=0.9\textwidth]{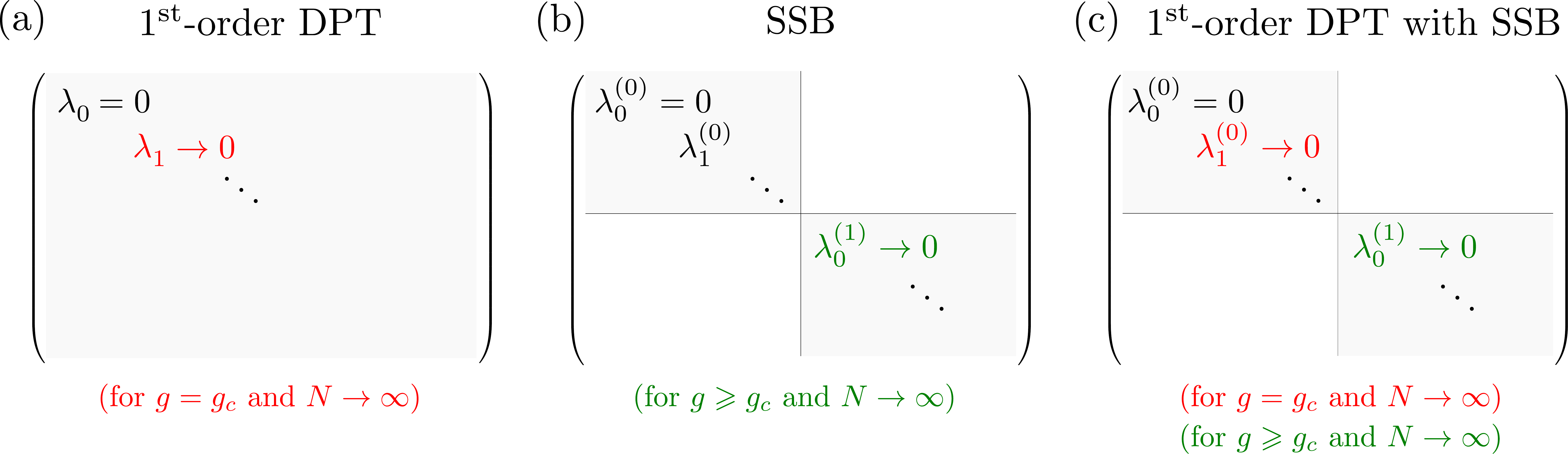}
    \caption{Main properties of the behavior of key eigenvalues of the HEOM Liouvillian in the thermodynamic limit $N \to \infty$ in different scenarii. (a) When there is a $1^{\mathrm{st}}$-order DPT, the first non-zero eigenvalue $\lambda_1$ should vanish at the critical point. (b) When the Liouvillian can be decomposed in symmetry sectors (illustrated here for the case of two symmetry sectors $k = 0$ and $k = 1$) and there is a SSB, the first non-zero eigenvalue of the Liouvillian in each symmetry sector different from the one of the steady state (i.e., $\lambda_0^{(k)}$ for $k > 0$) should vanish in the symmetry broken phase ($g \geqslant g_c$). (c) When there is both a $1^{\mathrm{st}}$-order DPT and a SSB, the system exhibits the combined behavior of cases (a) and (b), where the eigenvalue vanishing only at $g_c$ is the one in the symmetry sector of the steady state (i.e., $\lambda_1^{(0)}$).}
    \label{SM_FigSpectrum}
\end{figure*}

\section{Understanding the emergence of DPTs with spectral decompositions}

In this section, we illustrate how DPTs can be understood from spectral decomposition in the cases of the two models considered in the companion Letter~\cite{Debecker2024PRL}. 

\subsection{$1^{\mathrm{st}}$-order DPTs}
We have proven that the emergence of a $1^{\mathrm{st}}$-order DPT can be traced back to the existence of an eigenvalue $\lambda_1^{(0)}$ which vanishes at the critical point in the thermodynamic limit $N \rightarrow + \infty$ (see Proposition \ref{prop1}). Moreover, the associated eigenvector, namely $ \rho_{1}$ whose projection on the system sector $(\vec0,  \vec0)$ is written as $\rho_1^{(\vec0,  \vec0)}$, contains information about the system states right after/before the critical point, respectively denoted by $\rho_+^{(\vec0,  \vec0)}$ and $\rho_{-}^{(\vec0,  \vec0)}$. We also know that the steady state at the critical point can simply be written as a mixture with equal weights of the two phases $\rho_{\pm}^{(\vec0,  \vec0)}$ (see Propositions \ref{prop2} and \ref{prop3}). Those statements are strictly true in the thermodynamic limit $N \rightarrow + \infty$. We expect, however, that at finite $N$ the approximation 
\begin{equation}
    \rho_{ss}^{(\vec0,  \vec0)}(g=g_c, N) \approx \frac{\rho_+^{(\vec0,  \vec0)} (g_c, N) + \rho_-^{(\vec0,  \vec0)} (g_c, N)}{2},
    \label{rhoapprox}
\end{equation}
holds in a finite region around $g_c$ and improves in accuracy as $N$ increases. As in Sec.~\ref{sec::Spectral theory}, $\rho_{\pm}^{(\vec0,  \vec0)}(g_c, N)$ can be determined by spectral decomposition $\rho_1^{(\vec0,  \vec0)} (g_c, N) \propto \rho_+^{(\vec0,  \vec0)} (g_c, N) - \rho_-^{(\vec0,  \vec0)} (g_c, N)$. Furthermore, as we approach the thermodynamic limit, the states $\rho_{\pm}^{(\vec0,  \vec0)}(g_c, N)$ are expected to approach the states immediately before or after the critical point.

Here, we checked the statement above in the context of a dissipative Lipkin-Meshkov-Glick model that corresponds to Eq.~(\ref{Htot}) with $N_E = M_1 = 1$, with the system Hamiltonian
\begin{equation}
    H_S = \frac{V}{2N}\left(S_x^2 - S_y^2\right) = \frac{V}{2N}\left(S_+^2 + S_-^2\right),
    \label{H_LMG}
\end{equation}
and with the single jump operator $L = S_-$ and bath correlation function
\begin{equation}
    \alpha(\tau) = \frac{\gamma \kappa}{2N} e^{-\kappa |\tau| -i\omega \tau},
    \label{alpha_LMG}
\end{equation} 
where $S_\alpha = \sum_{j=1}^N \sigma_\alpha^{(j)}/2$ ($\alpha = x, y, z$) are collective spin operators defined in terms of single-spin Pauli operators $\sigma_\alpha^{(j)}$ and $S_\pm = S_x \pm i S_y$. The non-Markovian dynamics of the collective spin is thus governed by Eq.~(\ref{HEOM}) with $w \equiv w_{11} = \kappa + i \omega$ and $G \equiv G_{11} = \gamma \kappa/(2N)$. 

The model above exhibits a $1^{\mathrm{st}}$-order DPT as the parameter 
\begin{equation}
    g  \equiv \frac{V}{\gamma}
\end{equation} is varied. In the Markovian limit $\kappa \to \infty$ where $\alpha(\tau) \to (\gamma/N)\delta(\tau)$, i.e., when the bath has no memory, the critical point is $g_c^{M} = 1/2$~\cite{Lee2014D}, separating a steady state phase where 
$\langle S_z\rangle/(N/2) \to -1$ ($g < g_c^{M}$) to a phase where $\langle S_z\rangle/(N/2) \to 0$ ($g > g_c^{M}$) for $N \to \infty$. In the companion Letter, we show that deviations from the Markovian limit $\kappa \to \infty$ make it possible to shift toward lower values the critical point $g_c$.

In Fig.~\ref{SzSSmpasN}, we show the true quantum solution obtained through diagonalization of the HEOM generator and compare it with the states $\rho_{\pm}^{(\vec0,  \vec0)}$ and $(\rho_+^{(\vec0,  \vec0)} + \rho_-^{(\vec0,  \vec0)})/2$ as a function of $g$ and for two different values of $N$. As expected, the states $\rho_-^{(\vec0,  \vec0)}$ and $\rho_+^{(\vec0,  \vec0)}$ capture the right magnetization before and after the critical point, as shown in panels~(a) and~(c). More precisely, we see that $\rho_-^{(\vec0,  \vec0)}$ leads to $\langle S_z\rangle/(N/2) \approx -1$ (green curve) for $g < g_c$ and $\rho_+^{(\vec0,  \vec0)}$ leads to $\langle S_z\rangle/(N/2) \approx 0$ (blue curve) for $g > g_c$.

We also investigate the fidelity between $\rho_{ss}^{(\vec0,  \vec0)}$ and the states $\rho_{\pm}^{(\vec0,  \vec0)}$ and $(\rho_+^{(\vec0,  \vec0)} + \rho_-^{(\vec0,  \vec0)})/2$. For this, we use the definition 
\begin{equation}
    F(\rho, \sigma) = \mathrm{Tr}\sqrt{\sqrt{\rho}\sigma\sqrt{\rho}}
    \label{Fidelity_defintion}
\end{equation}
for the fidelity between two density operators $\rho$ and $\sigma$. We find that $\rho_-^{(\vec0,  \vec0)}$ ($\rho_+^{(\vec0,  \vec0)}$) is a good approximation to the true steady state $\rho_{ss}^{(\vec0,  \vec0)}$ in a region below (above) the critical point. Moreover, in a small region around the critical point, the steady state is best described by $(\rho_+^{(\vec0,  \vec0)} + \rho_-^{(\vec0,  \vec0)})/2 $. All these results are consistent with what we would expect for finite $N$; as $N$ increases, we find that the region in which $(\rho_+^{(\vec0,  \vec0)} + \rho_-^{(\vec0,  \vec0)})/2  \approx \rho_{ss}^{(\vec0,  \vec0)}$ is getting narrower [compare panels (b) and (d)], which is consistent with the fact that Eq.~\eqref{rhoapprox} strictly holds at a unique point in the thermodynamic limit. Finally, we note that the location of the critical point gets shifted as $N$ increases while the fidelities get higher.
\begin{figure}
    \centering
\includegraphics[width = 0.475\textwidth]{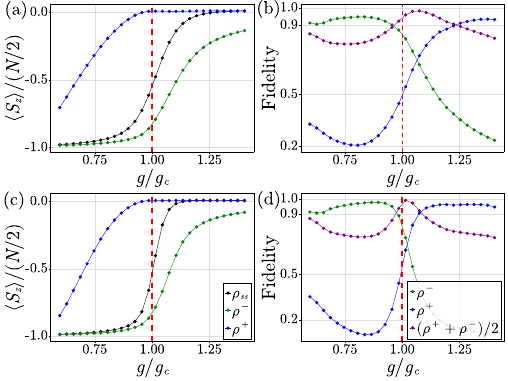}
    \caption{
Right panels: Fidelity between the steady states $\rho_{ss}^{(\vec0,  \vec0)}$ and the states $\rho_\pm^{(\vec0,  \vec0)}$ (green and blue curves) and $(\rho_+^{(\vec0,  \vec0)} + \rho_-^{(\vec0,  \vec0)})/2$ (purple curve) reconstructed from the spectral decomposition of $\rho_1^{(\vec0,  \vec0)}$. Left panels: Magnetization $\braket{S_z}/(N/2)$ as a function of $g$ for the steady state $\rho_{ss}^{(\vec0,  \vec0)}$ (black), $\rho_-^{(\vec0,  \vec0)}$ (green), $\rho_+^{(\vec0,  \vec0)}$ (blue). Panels (a) and (b) correspond to $N =30$ while panels (c) and (d) correspond to $N = 50 $. Parameters are $\kappa = \omega = \gamma$. The dashed red vertical line indicates the position of the critical point.}
    \label{SzSSmpasN}
\end{figure}

\subsection{$2^{\mathrm{nd}}$-order DPTs}
In this section, we illustrate the validity of the spectral decomposition~\eqref{rhoApprox2ndOrder}, namely 
\begin{equation}
 \rho_{ss} (g, N) \approx \frac{\rho_+^{(\vec0, \vec0)}(g, N) + \rho_-^{(\vec0, \vec0)}(g , N)}{2},
 \label{rhoApprox2ndOrderLMG}
\end{equation}
for the second model considered in the companion Letter, which corresponds to Eq.~(\ref{Htot}) with $N_E = M_1 = 1$, with the system Hamiltonian
\begin{equation}
    H_S = \frac{V}{2N}\left(S_+^2 + S_-^2\right) + h S_z,
    \label{H_LMG_2}
\end{equation}
with the jump operator $L = S_x$ and the bath correlation function
\begin{equation}
    \alpha(\tau) = \frac{\gamma \kappa}{2 N } e^{-\kappa |\tau| -i\omega \tau}.
\end{equation} 

The enlarged Markovian description of this model reads
\begin{equation}
    \dot \rho_\mathrm{tot} = -i[H, \rho_\mathrm{tot}] + \kappa \, \mathcal{D}_{a}[\rho_\mathrm{tot}]
   \label{master_tot}
\end{equation}
with $H = H_S + \omega a^\dagger a + \sqrt{G}(a^\dagger L+  L^\dagger a)$ with $G = \gamma \kappa/{2N}$, where $a$ is the annihilation operator of the single pseudomode of the enlarged Markovian embedding. From this, it is clear that the model exhibits a $\mathbb{Z}_2$ symmetry represented by $\mathcal{U}_2 = U_2 \bcdot U_2^\dagger$ with $U_2 = e^{i\pi (S_z + a^\dagger a)}$. The superoperator $\mathcal{U}_2$ has two distinct eigenvalues $u_k = e^{ik\pi}=\pm 1$ with $k=0, 1$, so there are two symmetry sectors, with $L_{u_0}$ containing the steady state. 

In the companion Letter, we show that the Markovian limit $\kappa \to \infty$ is not leading to any DPT upon varying the parameter $g = V/\gamma$, but that deviations from it yield two consecutive 2$^{\mathrm{nd}}$-order DPTs separating three different phases labeled (I), (II), and (III) at two critical points $g_{c,1}$ and $g_{c,2}$, with the $\mathbb{Z}_2$ symmetry spontaneously broken in phases (I) and (III). 

Hence, Eq.~\eqref{rhoApprox2ndOrderLMG} should hold in phases (I) and (III), for large but finite $N$. However, note that the explicit construction of $\rho_\pm^{(\vec0, \vec0)}$ requires $\lambda_0^{(1)}$ to be real such that the associated eigenvector is Hermitian. The imaginary part of this eigenvalue should vanish after the critical point in the thermodynamic limit because of the SSB (see Sec.~\ref{sec::Spectral theory}), however finite-size effects cause the imaginary part to vanish further than right after the critical point. Therefore, we illustrate the validity of Eq.~\eqref{rhoApprox2ndOrderLMG} only in the region where the imaginary part is negligible. Figure~\ref{Decomp_LMG_2ndOrder} highlights the accuracy of Eq.~\eqref{rhoApprox2ndOrderLMG} as shown by the high fidelity between the steady state obtained by numerical diagonalization of $\mathcal{L}_\mathrm{HEOM}$ and the reconstructed steady state $(\rho_+^{(\vec0, \vec0)} + \rho_-^{(\vec0, \vec0)})/2$. Although the $\mathbb{Z}_2$ symmetry is broken in phases (I) and (III), it is clear that the \textit{way} it is broken is different. Indeed, the convergence is slower and less smooth in phase (III) than in phase (I). This cannot be attributed to numerical errors due to a finite truncation order $k_\mathrm{max}$ since $\braket{a^\dagger a} = 0$ in phase (III)~\cite{Debecker2024PRL}. However, this could be due to the fact that in phase (I) it is the symmetry $S_x \rightarrow -S_x,~a \rightarrow -a$ that is broken, while in phase (III) it is the symmetry $S_y \rightarrow -S_y$. Interestingly, the former symmetry involves both the system and the environment, while the latter is more concerned with the coherent part of the generator. In this context, Fig.~\ref{Decomp_LMG_2ndOrder} suggests that the  breaking of the $\mathbb{Z}_2$ symmetry through involving both the system and the pseudomode is less prompt to undergo large finite-size effects.  
 
\begin{figure}
   \centering
\includegraphics[width = 0.475\textwidth]{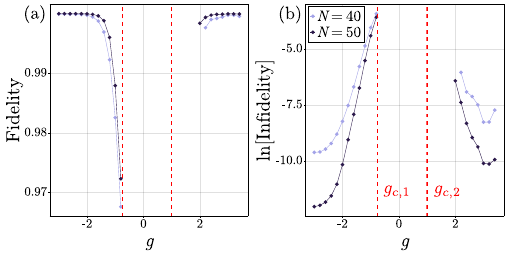}
   \caption{Fidelity (a) (see Eq.~\eqref{Fidelity_defintion}) and natural logarithm of the infidelity (b) (defined as $1$ - fidelity) between the reconstructed state~[Eq.~\eqref{rhoApprox2ndOrderLMG}] and the steady state $\rho_{ss}$ found by numerical diagonalization of $\mathcal{L}_\mathrm{HEOM}$ for the second model as a function of $g$ in phases (I) and (III) for  $N = 40, 50$. Panels (a) and (b) prove the validity of the decomposition~\eqref{rhoApprox2ndOrderLMG}. Note that we only show the region for which $\lambda_0^{(1)}$ is real ($\mathrm{Im}[\lambda_0^{(1)}]/\omega < 10^{-8}$). Parameters are $\omega = \kappa = 2 \gamma = 2 h$ and $k_\mathrm{max} = 9$, which yields $g_{c,1} = -3/4$ and $g_{c,2} = 1$.}
   \label{Decomp_LMG_2ndOrder}
\end{figure}

\section{Challenging two-mode Dicke model}

We now provide an example of a DPT that appears in a two-mode Dicke model described by the Lindblad master equation~\cite{Moodie2018, Palacino2020}
\begin{equation}
    \dot{\rho}_\mathrm{tot} = -i[H_\mathrm{tot}, \rho_\mathrm{tot}] + \kappa (\mathcal{D}_a[\rho_\mathrm{tot}] + \mathcal{D}_b[\rho_\mathrm{tot}]),
    \label{U1-ME}
\end{equation}
where $a, b$ are bosonic annihilation operators damped at rate $\kappa$ and
where
\begin{equation}
    H_\mathrm{tot} = \omega_0 S_z + \omega ( a^\dagger a +  b^\dagger b) + \frac{g}{\sqrt{N}}(a S_+ + b S_- + \mathrm{h.c.}).
    \label{U1-HAC}
\end{equation}
This model corresponds to a collective spin system coupled to two damped cavity modes, that is Eq.~(\ref{Htot}) with $N_E = 2$,  $M_1 = M_2 = 1$, the system Hamiltonian
\begin{equation}
    H_S = \omega_0 S_z,
    \label{HDicke}
\end{equation}
the jump operators $L_1 = S_-$, $L_2 = S_+$ and the bath correlation functions
\begin{equation}
    \alpha_1(\tau) = \alpha_2(\tau) = \frac{g^2}{N} e^{-\kappa |\tau| -i\omega \tau}.
\end{equation} 
The model~\eqref{U1-ME} exhibits a continuous $\mathbb{U}(1)$ symmetry described by $\mathcal{U}_1 = U_1 \bcdot U_1^\dagger$ with $U_1 = e^{i\alpha(S_z+a^\dagger a- b^\dagger b)}$ $(\alpha \in \mathbb{R})$. It undergoes a 2$^{\mathrm{nd}}$-order DPT associated with a breaking of the $\mathbb{U}(1)$ symmetry at the critical point
\begin{equation}
    g_c = \sqrt{\frac{\omega_0(\omega^2 + \kappa^2)}{2\omega}}.
\end{equation}
This critical point separates a normal phase with $|\braket{S_z}|/(N/2)= 1 $ for $g < g_c$ to a superradiant phase (where the symmetry is broken) with $|\braket{S_z}|/(N/2) < 1$ for $g > g_c$ as $N \to \infty$~\cite{Moodie2018}.

This model offers a challenge for two main reasons. First, because there are two pseudomodes to consider in the enlarged system made of the spin and the pseudomodes, which increases drastically the size of the corresponding Hilbert space while in the same time motivates the search for appropriate reduced descriptions of the spin dynamics. Second, as explained in the introduction, such reduced descriptions of spin dynamics have been studied and compared in~\cite{Palacino2020} with the mean-field results summarized above. However, so far, none of them were able to describe all aspects of the phase transition. Indeed, it has been shown that, by contrast with the single-mode Dicke model~\cite{Damanet2019}, a standard Redfield approach completely misses the DPT, while a 4$^{\mathrm{th}}$-order Redfield master equation (i.e., a 4$^\mathrm{th}$-order perturbative treatment of the interaction) captures the correct steady state and critical point but fails to predict the closing of the gap, a necessary condition for DPT. This suggests that the correct description of the low-lying spectrum of the generator of the dynamics requires including higher-order corrections in the interactions, i.e., deepen more in the non-Markovian regime. All of this makes this model an ideal candidate for testing and demonstrating the power of our general framework. In fact, our method captures all the features of the DPT and SSB, as shown in Fig.~\ref{2dOrder}, which shows the magnetization $\langle S_z \rangle$ (a), the closing of the gap $|\mathrm{Re}[\lambda_0^{(k >0)}]|$ (c,d) and of the imaginary part of $\lambda_0^{(k > 0)}$ (b).

\begin{figure}
    \centering
\includegraphics[width = 0.475\textwidth]{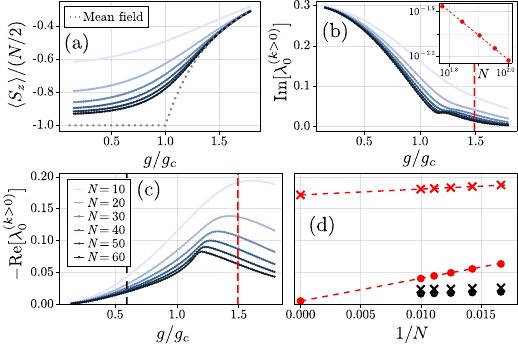}
    \caption{Signatures of the 2$^\mathrm{nd}$-order DPT for the two-mode Dicke model~(\ref{U1-HAC}) obtained from $\mathcal{L}_{\mathrm{HEOM}}$ for $\kappa = \omega = 5\omega_0$. (a): Steady state magnetization $\braket{S_z}$ as a function of $g$ for $N = 10$-$50~(k_\mathrm{max} = 7)$ and $60~(k_\mathrm{max} = 8)$. As $N$ increases, the curves get closer to the mean-field result (dotted line). (b) Imaginary part of the eigenvalue $\lambda_0^{(k>0)}$ with the largest real part among all symmetry sectors with $k >0$ as a function of $g/g_c$, confirming the SSB. The inset shows the scaling of $\textrm{Im}[\lambda_0^{(k>0)}]$ as a function of $N$ at $g/g_c = 1.49$. (c) $-\mathrm{Re}[\lambda_0^{(k>0)}]$ as a function of $g/g_c$ showing a decreasing gap in the superradiant phase as $N$ increases. The vertical dashed lines show $g/g_c = 1.49$ (red) and $g/g_c = 0.6$ (black) used in panel (d) to compare the scaling of $-\mathrm{Re}[\lambda_0^{(k>0)}]$ (circles) and of the Liouvillian gap of the $4^{\mathrm{th}}$ Redfield master equation of~\cite{Palacino2020} (crosses) as a function of $1/N$ [with $N = 60$-$90~(k_\mathrm{max} = 8$) and $100~(k_\mathrm{max} = 9)$]. In the normal phase (black), both methods are in good agreement, while in the superradiant phase (red), only $\mathcal{L}_{\mathrm{HEOM}}$ gives the expected closing. The points at $1/N = 0$ were extrapolated from a line defined by the two last points of our data.}
    \label{2dOrder}
\end{figure}

\section{Conclusion}
We developed a comprehensive framework for studying dissipative phase transitions in non-Markovian systems, which are more relevant experimentally. Our method is numerically exact, systematic, easily accessible (since it is based on the well-established HEOM technique available in open access libraries~\cite{Lambert2020Qutip, Huang2023heomjl}), and provides a considerable computational advantage over a standard embedding technique. 
As an example, we applied our theory to capture all the defining features of a challenging 2$^{\mathrm{nd}}$-order DPT with a continuous SSB for which other previous reduced descriptions had failed up to now~\cite{Palacino2020}. In our companion Letter, we showed that non-Markovian effects can not only reshape phase boundaries but also trigger phase transitions, illustrating the power and the opportunities that our framework offers for future works.

While our work focuses on $1^\mathrm{st}$ and $2^\mathrm{nd}$-order phase transitions, it would be interesting to determine the signatures of infinite order phase transitions, such as BKT transitions, in the spectrum of the HEOM Liouvillian. For such phase transitions, under the hypotheses of our Propositions \ref{prop1}-\ref{prop3}, the gap should not close only at a particular point, as this would imply a discontinuity of the order parameter. This is consistent with what is known in the literature for closed systems~\cite{Itzykson1989}, i.e., that the gap closes exponentially in a certain region around the critical point.
Another perspective could be to improve our method via hybridization with advanced numerical techniques, such as corner-space renormalization~\cite{Finazzi2015Corner} or matrix product operators (as in~\cite{Mascarenhas2015Matrix, Cui2015Variational, Flannigan2021, Moroder2022Strongly,link2023open}). Moreover, it would be interesting to include initial system-bath correlations~\cite{Ikeda2020Generalization}, but also the case of non-Lorentzian environments, i.e., where the correlation functions do not correspond to decaying exponentials and thus where Markovian embedding is not possible,  via the extended HEOM (eHEOM) technique~\cite{Tang2015Extended}. This would open opportunites for the study of critical behaviors~\cite{ Nagy2016Critical,Lundgren2020Nature, Ikeda2020Generalization}. 
Other interesting research directions include investigating whether non-Markovianity could be used as a tool to produce steady states that cannot be reached by Markovian dissipation alone~\cite{Ask2022} or to transform quasi-equilibrium steady state of dissipative systems (appearing, e.g., in cavity QED~\cite{Torre2013}) into genuine out-of-equilibrium ones.

\begin{acknowledgments}
We thank Jonathan Keeling, Peter Kirton, Valentin Link and Lukas Pausch for
helpful comments on a previous version of the manuscript. Computational resources were provided by the Consortium des Equipements de Calcul Intensif (CECI), funded by the Fonds de la Recherche Scientifique de Belgique (F.R.S.-FNRS) under Grant No. 2.5020.11 and by the Walloon Region. 
\end{acknowledgments}

\appendix

\section{Explicit matrix form of the HEOM Liouvillian}
\label{appExplicitform}
In this section, we show how one can construct the matrix representation of the HEOM Liouvillian. By vectorizing Eq.~\eqref{HEOM} thanks to the Choi-Jamiolkowski ismorphism ($\ket i\bra j\cong \ket i \otimes \ket j$~\cite{Zwolak2004}), we get
\begin{equation}
\label{eqrhonm}
 \begin{aligned}
    \dot{\dket{\rho^{(\vec{n}, \vec{m})}}}\;=&  -i\left[H_S\otimes \mathds{1} - \mathds{1}\otimes H_S^T - (\vec{w}^*\bcdot\vec{n} + \vec{w}\bcdot\vec{m})\right] \dket{\rho^{(\vec{n}, \vec{m})}} \\
    &+ \sum_{l=1}^{N_E}\sum_{j = 1}^{M_l} \left(G_{lj} n_{lj} L_l\otimes \mathds{1} \dket{\rho^{(\vec{n} -\vec{e}_{lj},\vec{m})}} \right.  \nonumber \\
    &\qquad\qquad+ \left. G_{lj}^* m_{lj} \mathds{1}\otimes L_l^* \dket{\rho^{(\vec{n}, \vec{m}-\vec{e}_{lj})}} \right)  \\
    &+ \sum_{l=1}^{N_E}\sum_{j = 1}^{M_l} \left[(\mathds{1}\otimes L_j^* - L_j^\dagger \otimes \mathds{1}) \dket{\rho^{(\vec{n}+\vec{e}_{lj}, \vec{m})}}\right. \\
     &\qquad\qquad+\left.(\mathds{1}\otimes L_l^* - L_l^\dagger \otimes \mathds{1})^\dagger \dket{\rho^{(\vec{n}, \vec{m}+\vec{e}_{lj})}} \right],
     \end{aligned}
\end{equation}
where $\dket{\rho^{(\vec{n}, \vec{m})}}$ denotes the vectorization of the matrices $\rho^{(\vec{n}, \vec{m})}$, $\mathds{1}$ the identity matrix acting on $\mathcal{H}_S$, and $L_l^*$ ($L_l^T$) the conjugate (transpose) matrix of $L_l$. 
By stacking in a vector $\dket{\rho}$ all the vectorized matrices $\dket{\rho^{(\vec{n}, \vec{m})}}$, we can construct the matrix $\mathcal{L}_{\mathrm{HEOM}}(k_{\mathrm{max}})$ and the equation above becomes Eq.~\eqref{def_LHEOM}, i.e.,
\begin{equation}
    \frac{d\dket{\rho}}{dt} = \mathcal{L}_{\mathrm{HEOM}}(k_{\mathrm{max}})\,\dket{\rho}.
    \label{def_Leff}
\end{equation}

For the sake of clarity, we now explicitly construct the different blocks of the matrix representation of the HEOM Liouvillian
for a unique environment $N_E = 1$ made of one damped pseudomode ($M = 1$). For this special case, the vectorized HEOM reads
\begin{equation}
 \begin{aligned}
   \dot{\dket{\rho^{(n, m)}}}
    ={} & D_{nm} \dket{\rho^{(n, m)}} \nonumber 
    +  A_n\dket{\rho^{(n-1,m)}}  + B_m\dket{\rho^{(n, m-1)}} \nonumber\\
    &+ C \dket{\rho^{(n+1, m)}} - C^\dagger
     \dket{\rho^{(n, m+1)}},
 \end{aligned}
 \label{SHEOM_1mode}
\end{equation}
where 
\begin{equation}
   \begin{split}
      D_{nm} &\equiv -i(H_S\otimes \mathds{1} - \mathds{1}\otimes H_S^T) \\
      &- ((n-m)i\omega + (n+m) \kappa) \mathds{1}\otimes \mathds{1}, \\
      A_n &\equiv G n~L\otimes \mathds{1},  
       \quad B_m \equiv G^* m~ \mathds{1}\otimes L^*, \\
       C &\equiv \mathds{1}\otimes L^* - L^\dagger \otimes \mathds{1}.
   \end{split} 
\end{equation}
Therefore, if $k_\mathrm{max} = 1$, the stacked vector $\dket{\rho}$ is given by $\dket{\rho} = (\dket{\rho^{(0, 0)}}, \dket{\rho^{(0, 1)}}, \dket{\rho^{(1, 0)}})^T$ and 
\begin{equation}
    \mathcal{L}_\mathrm{HEOM}(k_\mathrm{max} = 1) = 
    \begin{pmatrix}
    D_{00} & -C^\dagger & C \\
    B_{1} & D_{01} & 0 \\
    A_1 & 0 & D_{10}
    \end{pmatrix},
\end{equation}
while for $k_\mathrm{max} = 2$, we get $\dket{\rho} = (\dket{\rho^{(0, 0)}}, \dket{\rho^{(0, 1)}}, \dket{\rho^{(0, 2)}}, \dket{\rho^{(1, 0)}}, \dket{\rho^{(1, 1)}}, \dket{\rho^{(2, 0)}})^T$ and 
\begin{equation}
    \mathcal{L}_\mathrm{HEOM}(2) = 
    \begin{pmatrix}
    D_{00} & -C^\dagger & 0 & C & 0 & 0 \\
    B_1 & D_{01} & -C^\dagger & 0 & C& 0 \\
    0 & B_2 & D_{02} & 0 & 0 & 0 \\
    A_1 & 0 & 0 & D_{10} & -C^\dagger & C \\
    0 & A_1 & 0 & B_1 & D_{11} & 0 \\
    0 & 0 & 0 & A_2 & 0 & D_{20}
    \end{pmatrix}.
\end{equation}

\section{Proofs of the spectral properties of $\mathcal{L}_\mathrm{HEOM}$}
\label{Spectral_Proofs}
We provide here the proofs of the properties of the HEOM Liouvillian. To show that the spectrum of $\mathcal{L}_\mathrm{HEOM}$ is symmetric with respect to the real axis (property (i)), we note that $
   \left(\frac{d\rho^{(\vec{n}, \vec{m})}}{dt}\right)^\dagger = \frac{d}{dt}\rho^{(\vec{m},  \vec{n})} = \frac{d}{dt}(\rho^{(\vec{n},  \vec{m})})^\dagger$, where we used the property $(\rho^{(\vec{n}, \vec{m})})^\dagger = \rho^{(\vec{m}, \vec{n})}$~\cite{Link2022}, which
implies 
\begin{equation}
    \mathcal{L}_\mathrm{HEOM}[\rho^\dagger] = (\mathcal{L}_\mathrm{HEOM}[\rho])^\dagger.
    \label{Symmetry_real_axis}
\end{equation} 
The trace preserving property ((ii)) of $\mathcal{L}_\mathrm{HEOM}$ is immediate from Eq.~(\ref{HEOM}).This implies that 
\begin{equation}
    \begin{split}
   0 &{}=\frac{d\mathrm{Tr}[\rho^{(\vec{0}, \vec{0})}]}{dt} = \mathrm{Tr}\left[
    \frac{d}{dt}\rho^{(\vec{0}, \vec{0})}\right] = \mathrm{Tr}\left[\mathds{1}^{(\vec{0}, \vec{0})}~\mathcal{L}_{\mathrm{HEOM}}[\rho] \right] \\& =\dbraket{\mathds{1}^{(\vec{0}, \vec{0})}|\mathcal{L}_{\mathrm{HEOM}}|\rho}  \quad \forall\;\rho,
\label{0eigenvalue}
\end{split}
\end{equation}
where we used the Hilbert-Schmidt inner product $\dbraket{A|B} \equiv \mathrm{Tr}[A^\dagger B]$ and the projector onto the physical state space $\mathds{1}^{(\vec{0}, \vec{0})}$. Equation~\eqref{0eigenvalue} leads to $\dbra{\mathds{1}^{(\vec{0}, \vec{0})}}\mathcal{L}_{\mathrm{HEOM}} =0$, meaning that $\dbra{ \mathds{1}^{(\vec{0}, \vec{0})}}$ is a left eigenvector of $\mathcal{L}_{\mathrm{HEOM}}$ associated to the eigenvalue $0$. Therefore, the eigenvalue $0$ is always in the spectrum of $\mathcal{L}_{\mathrm{HEOM}}$ (property (iii)), which guarantees the existence of a stationary state. The fact that all the eigenvalues must have a negative real part in the limit $k_{\mathrm{max}} \rightarrow +\infty$ (property (iv)) comes from the fact that in this limit, the solution of Eq.~\eqref{HEOM} of the main text in the sector $(\vec{0}, \vec{0})$ is \textit{exactly} the reduced density operator of the system. Thus, any positive real part eigenvalues would lead to unphysical matrices in the sector $(\vec{0}, \vec{0})$, therefore contradicting our last statement. Lastly, to prove that $\mathrm{Tr}[\mathds{1}^{(\vec{0}, \vec{0})}\rho_i] = 0$ if $\rho_i$ is a right eigenoperator of $\mathcal{L}_{\mathrm{HEOM}}$ associated to the eigenvalue $\lambda_i$ with $\mathrm{Re}[\lambda_i] \neq 0$ (property (v)), we note that $\mathcal{L}_{\mathrm{HEOM}}$ preserves the trace in the sector $(\vec{0}, \vec{0})$ and $\rho_i(t) = e^{\mathcal{L}_{\mathrm{HEOM}}t}\rho_i \rightarrow 0$ for $t\rightarrow +\infty$ if $\mathrm{Re}[\lambda_i] \neq 0$ and $k_{\mathrm{max}} \rightarrow +\infty$.

\section{Convergence analysis and numerical efficiency} \label{Numerical_Efficiency}
The only parameter relevant to the convergence analysis of $\mathcal{L}_\mathrm{HEOM}$ is the truncation order $k_{\mathrm{max}}$. We introduce the following measures of convergence  
\begin{equation}
 \begin{split}
  &C_{k_\mathrm{max}}(O) \equiv \left|\mathrm{Tr}\left[\rho_{ss}(k_\mathrm{max})O - \rho_{ss}(k_\mathrm{max}+1)O\right]\right|, \\
  &S_{k_\mathrm{max}}(\lambda) \equiv |\lambda(k_\mathrm{max})-\lambda(k_{\mathrm{max}+1})|,
  \end{split}
  \label{measuresOfConvergence}
\end{equation}
to assess the convergence of $\mathcal{L}_{\mathrm{HEOM}}$ with respect to the steady state expectation value of a given operator $O$ or with respect to one of its eigenvalue $\lambda$, such as the HEOM Liouvillian gap. Here, $\rho_{ss}(k_\mathrm{max})$ is the steady state  of $\mathcal{L}_\mathrm{HEOM}(k_{\mathrm{max}})$ and similarly $\lambda(k_\mathrm{max})$ is $\lambda$ computed with $\mathcal{L}_\mathrm{HEOM}(k_\mathrm{max})$. Note that the convergence measure $C_{k_{\mathrm{max}}}(O)$ is a natural choice often chosen to study the convergence of hierarchy of equations~\cite{HOPS_Suess14}. The first part of this section is dedicated to the convergence analysis of $\mathcal{L}_\mathrm{HEOM}(k_\mathrm{max})$ while the second part shed light on the numerical advantage of $\mathcal{L}_\mathrm{HEOM}$ over enlarged Markovian systems. 
\subsection{Convergence analysis of $\mathcal{L}_\mathrm{HEOM}(k_\mathrm{max})$}
In Fig.~\ref{LMG_convergence}, we show the two measures of convergence~\eqref{measuresOfConvergence} for the LMG model for $O = S_z$ and $\lambda = \lambda_0^{(1)}$. As the hierarchy depth $k_\mathrm{max}$ increases, both measures of convergence $C_{k_\mathrm{max}}(S_z)$ and $S_{k_\mathrm{max}}(\lambda_0^{(1)})$ globally decrease, showing that the truncation order $k_\mathrm{max}$ can be used to control the numerical errors inherent to the $\mathcal{L}_\mathrm{HEOM}(k_\mathrm{max})$ scheme. A comparison of the panels (a) and (c) with the panels (b) and (d) indicates that errors scale up as $N$ increases. We also note that it is numerically more challenging to extract the spectral quantity $\lambda_0^{(1)}$ than the steady state expectation value $\braket{S_z}$, as indicated by the change in scale on the y-axis between panels (a) and (c) and (b) and (d). 

These general observations still hold for the $\mathbb{U}(1)$-symmetric Dicke model, as illustrated in Fig~\ref{U1Dicke_convergence}, which is the analog of Fig~\ref{LMG_convergence} but for the $\mathbb{U}(1)$-symmetric Dicke model. We note that both $C_{k_\mathrm{max}}(S_z)$ and $S_{k_\mathrm{max}}(\lambda_0^{(k>0)})$ increases as the coupling $g$ increases, highlighting the numerical challenge of the so-called strong coupling regime. Moreover, this observation combined with the fact that the computation of $\lambda_0^{(k>0)}$ is numerically more demanding that of $\braket{S_z}$ could explain why a fourth order Redfield master equation seems to capture the right steady state but predicts a non-vanishing gap~\cite{Palacino2020}. We indeed foresee that the spectrum of $\mathcal{L}_\mathrm{HEOM}$ converges faster for larger eigenvalues.
\begin{figure}
\centering\includegraphics[width = 0.475\textwidth]{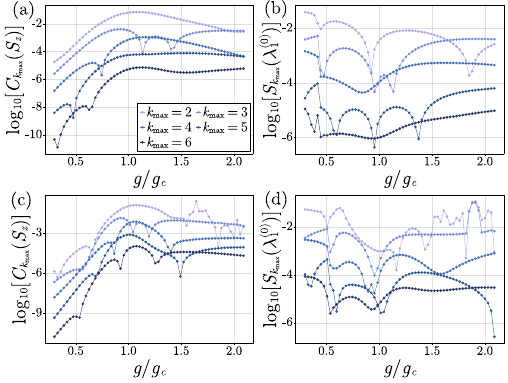}
    \caption{Measures of convergence $C_{k_\mathrm{max}}(S_z)$ and $S_{k_\mathrm{max}}(\lambda_1^{(0)})$ as defined by Eq.~\eqref{measuresOfConvergence} for the LMG model [Eqs.~(\ref{H_LMG}-\ref{alpha_LMG})] displayed in logarithmic scale (base 10) as a function of $g$. For all plots, the parameters are $\kappa = \omega = \gamma$ and $N=10$ for panels (a) and (c) and $N=20$ for panels (b) and (d).}
    \label{LMG_convergence}
\end{figure}
\begin{figure}
    \centering\includegraphics[width=0.475\textwidth]{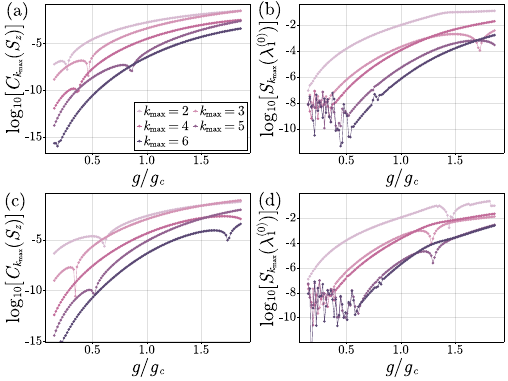}
    \caption{Measures of convergence $C_{k_\mathrm{max}}(S_z)$ and $S_{k_\mathrm{max}}(\lambda_0^{(k >0)})$ as defined by Eq.~\eqref{measuresOfConvergence} for the $\mathbb{U}(1)$-symmetric Dicke model discussed in Sec.~V displayed in logarithmic scale (base 10) as a function of $g$. For all panels, the parameters are $\kappa = \omega = 5\omega_0$ and $N=10$ for panels (a) and (b) and $N = 20$ for panels (c) and (d).}
    \label{U1Dicke_convergence}
\end{figure}

\subsection{Comparison with enlarged Markovian systems}

Let us illustrate the numerical advantage of our method to characterize DPTs over the standard technique of analysing the spectrum of the Liouvillian for the LMG model defined by Eqs.~(\ref{H_LMG}-\ref{alpha_LMG}). For this model, the Markovian Liouvillian superoperator $\mathcal{L}_M$ is defined through
\begin{equation}
\dot\rho_{\mathrm{tot}} = - i \left[H,\rho_\mathrm{tot}\right] + \kappa\left(2 a \rho_\mathrm{tot} a^\dagger - \{ a^\dagger a, \rho_\mathrm{tot}\} \right) \equiv \mathcal{L}_M[\rho_\mathrm{tot}].
\end{equation}
where $H = H_\mathrm{LMG} + \omega a^\dagger a + \sqrt{\frac{\gamma \kappa}{2N}}(S_- a^\dagger + S_+ a)$. 
As the dimension of $\mathcal{L}_M$ is infinite, one has to introduce a cutoff in order to determine the steady state of $\mathcal{L}_M$ numerically. We denote by $N_c$ and $\mathcal{L}(N_c)$ the effective dimension of the truncated Fock space of the pseudomode and the associated truncated Markovian Liouvillian. In order to compare $\mathcal{L}_M$ and $\mathcal{L}_\mathrm{HEOM}$, we fix a threshold of tolerance for the measures of convergence, namely $\epsilon = 0.0001$. We then choose $k_\mathrm{max}$ and $N_c$ accordingly: we take the first value  of $k_\mathrm{max}$ and $N_c$ that satisfy $ C_{k_\mathrm{max}}(S_z) < \epsilon$ and 
\begin{equation}
 C_{N_c}(S_z) \equiv \left|\tr(\rho_{ss}(N_c)S_z - \rho_{ss}(N_c+1)S_z)\right| < \epsilon,
\end{equation}
where $\rho_{ss}(N_c)$ is the steady state associated with $\mathcal{L}_M(N_c)$. We then compute the effective dimension of $\mathcal{L}_M$ and $\mathcal{L}_\mathrm{HEOM}$ for the truncation parameters $k_\mathrm{max}$ and $N_c$ previously determined. Figure~\ref{ComparisonLM/LNM}(a) shows that the ratio $\mathrm{dim}(\mathcal{L}_\mathrm{HEOM})/\mathrm{dim}(\mathcal{L}_M)$ is below 0.4 for all $g = V/\gamma$ and $N$ considered. Moreover, this ratio decreases with $N$, which shows that the $\mathcal{L}_\mathrm{HEOM}$ scheme is more suited for the study of DPTs for which one must consider the thermodynamic limit $N \rightarrow +\infty$. Let us finally mention that the generators $\mathcal{L}_\mathrm{HEOM}$ and $\mathcal{L}_M$ give the same results at the chosen tolerance threshold as illustrated in Fig.~\ref{ComparisonLM/LNM}(b).

\begin{figure}[h]
\centering\includegraphics[width = 0.475\textwidth]{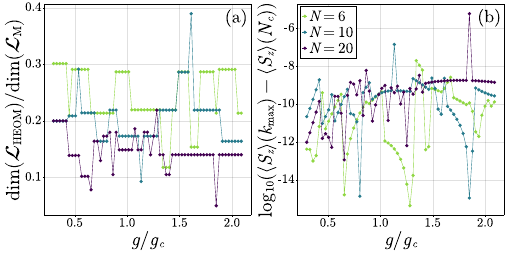}
    \caption{Comparison of the convergence of $\mathcal{L}_\mathrm{HEOM}$ and $\mathcal{L}_\mathrm{M}$ for the LMG model [Eqs.~(\ref{H_LMG}-\ref{alpha_LMG})]. (a): Ratio between the dimension of the HEOM generator $\mathcal{L}_\mathrm{HEOM}$ and the Markovian one $\mathcal{L}_\mathrm{M}$ as a function of $g$ for $\epsilon = 0.0001$, proving the numerical gain of using $\mathcal{L}_\mathrm{HEOM}$ instead of $\mathcal{L}_\mathrm{M}$ for the enlarged Markovian system. (b): Differences in logarithmic scale (base 10) between the steady state expectation value $\braket{S_z}$ computed with $\mathcal{L}_\mathrm{HEOM}(k_\mathrm{max})$ [resp.\ $\mathcal{L}_\mathrm{M}(N_c)$] denoted by $\braket{S_z}(k_\mathrm{max}$) [resp.\ $\braket{S_z}(N_c)$] for $\epsilon = 0.0001$. The two methods are in good agreement at the given tolerance.}
    \label{ComparisonLM/LNM}
\end{figure}

\bibliographystyle{apsrev4-2}
\bibliography{bib}

\begin{thebibliography}{84}
\bibitem [{foo()}]{footnote_multiindices}%
  \BibitemOpen
  \href@noop {} {}\bibinfo {note} {More explicitely, we mean $\vec{n} = (n_{11}, \dotsc, n_{1M_1}, n_{22}, \dotsc,n_{2M_2},\dotsc, n_{N_E1}, \dotsc, n_{N_EM_{N_E}})$ where $n_{lj} \in \mathbb{N}^M$.}\BibitemShut {Stop}%
\end{thebibliography}

\end{document}